\documentclass{l4dc2024}

\usepackage{amssymb}
\usepackage{scalerel}
\usepackage{lipsum}
\usepackage{amsfonts,amssymb}

\usepackage{graphicx}
\usepackage{epstopdf}
\usepackage{dsfont}
\usepackage{bbm}
\usepackage{indentfirst}
\usepackage{hhline}
\usepackage{mathtools}
\usepackage{commath}
\usepackage{extarrows}
\usepackage{mathrsfs} 
\usepackage{lineno}
\usepackage{multirow}

\usepackage{scalerel}
\newcommand\scale[2]{\vstretch{#1}{\hstretch{#1}{#2}}}

\usepackage{textcomp}
\usepackage{stfloats}
\usepackage{pgfplots}
\pgfplotsset{compat=1.17} 
\usepackage{stmaryrd}
\usepackage{algorithm}
\usepackage[noend]{algpseudocode}
\usepackage[colorinlistoftodos]{todonotes}

\allowdisplaybreaks

\makeatletter
\newcommand\multiline[1]{\parbox[t]{\dimexpr\linewidth-\ALG@thistlm}{#1}}
\makeatother

\usepackage{tikz}
\usetikzlibrary{arrows.meta,snakes,shapes,decorations.pathmorphing,backgrounds,positioning,fit,petri}
\usetikzlibrary{decorations,decorations.markings} 
\usetikzlibrary{positioning,arrows}
\usetikzlibrary{calc}

\usepackage{centernot}
\usepackage{hyperref}
\hypersetup{colorlinks=true, linkcolor=blue, breaklinks=true, urlcolor=blue}

\ifpdf
  \DeclareGraphicsExtensions{.eps,.pdf,.png,.jpg}
\else
  \DeclareGraphicsExtensions{.eps}
\fi

\newtheorem{defi}{\textbf{Definition}}

\newtheorem{rek}{\textbf{Remark}}

\newtheorem{lema}{\textbf{Lemma}}

\newtheorem{cor}{\textbf{Corollary}}

\newcounter{exampcount}
\newcommand{\algoref}[1]{Algorithm~\ref{#1}}

\newcommand{\lemaref}[1]{Lemma~\ref{#1}}

\newcommand{\obs}{\mathit{obs}}

\newcommand{\Loc}{\mathit{Loc}}
\newcommand{\Per}{\mathit{Per}}

\newcommand{\loc}{\mathit{loc}}
\newcommand{\per}{\mathit{per}}
\newcommand{\csg}{\mathsf{C}}

\newcommand{\agent}{\mathsf{Ag}}

\newcommand{\fpaths}{\mathit{FPaths}}

\pgfkeys{/tikz/.cd,
    street mark distance/.store in=\StreetMarkDistance,
    street mark distance=10pt,
    street mark step/.store in=\StreetMarkStep,
    street mark step=1pt,
}

\pgfdeclaredecoration{street mark}{initial}
{% 
\state{initial}[width=\StreetMarkStep,next state=cont] {
    \pgfmoveto{\pgfpoint{\StreetMarkStep}{\StreetMarkDistance}}
    \pgfpathlineto{\pgfpoint{0.3\pgflinewidth}{\StreetMarkDistance}}
    \pgfcoordinate{lastup}{\pgfpoint{1pt}{\StreetMarkDistance}}
    
  }
  \state{cont}[width=\StreetMarkStep]{
     \pgfmoveto{\pgfpointanchor{lastup}{center}}
     \pgfpathlineto{\pgfpoint{\StreetMarkStep}{\StreetMarkDistance}}
     \pgfcoordinate{lastup}{\pgfpoint{\StreetMarkStep}{\StreetMarkDistance}}
  }
  \state{final}[width=\StreetMarkStep]
  {
    \pgfmoveto{\pgfpointdecoratedpathlast}
  }
}

\newcommand{\startpara}[1]{{%
\vskip4pt\noindent
{\bf #1.}}}

\makeatletter
\renewcommand{\ALG@name}{\sc Algorithm}
\makeatother

\usepackage{amsopn}

\title[Short Title]{HSVI-based Online Minimax Strategies for Partially Observable Stochastic Games with Neural Perception Mechanisms}
\usepackage{times}

\coltauthor{\Name{Rui Yan}$^1$ \Email{rui.yan@cs.ox.ac.uk}\\
 \Name{Gabriel Santos}$^1$ \Email{gabriel.santos@cs.ox.ac.uk}\\
 \Name{Gethin Norman}$^{1,2}$ \Email{gethin.norman@glasgow.ac.uk}\\
 \Name{David Parker}$^1$ \Email{david.parker@cs.ox.ac.uk}\\
 \Name{Marta Kwiatkowska}$^1$ \Email{marta.kwiatkowska@cs.ox.ac.uk}\\
 \addr $^1$Department of Computer Science, University of Oxford, Oxford, OX1 3QD, UK\\
 \addr $^2$School of Computing Science, University of Glasgow, Glasgow, G12 8QQ, UK}

\begin{document}

\maketitle

\begin{abstract}%
 We consider a variant of continuous-state partially-observable stochastic games
 with neural perception mechanisms and an
 asymmetric information structure. One agent has partial information, with the observation function implemented as a neural network, while the other agent is assumed to have full knowledge of the state. We present, for the first time, an efficient online method to compute an $\varepsilon$-minimax strategy profile, which requires only one linear program to be solved for each agent at every stage, instead of a complex estimation of opponent counterfactual values.  
 For the partially-informed agent, we propose a continual resolving approach which uses lower bounds, pre-computed offline with heuristic search value iteration (HSVI), instead of opponent counterfactual values. This inherits the soundness of continual resolving at the cost of pre-computing the bound. For the fully-informed agent, we propose an inferred-belief strategy, where the agent maintains 
an inferred belief about the belief of the partially-informed agent based on (offline) upper bounds from HSVI, guaranteeing $\varepsilon$-distance to the value of the game at the initial belief known to both agents. 
\end{abstract}


\section{Introduction}
\noindent
Partially-observable stochastic games (POSGs) are a modelling formalism that enables strategic reasoning and (near-)optimal synthesis of strategies and equilibria in multi-agent settings with partial observations and uncertainty. 
\emph{One-sided POSGs}~\citep{KH-BB-VK-CK:23} are a tractable subclass of two-agent, zero-sum POSGs with an asymmetric information structure, where only one agent has partial information while the other agent is assumed to have full knowledge. This is well suited to autonomous safety- or security-critical settings, such as patrolling or pursuit-evasion games, which require reasoning about worst-case assumptions.
Since real-world
settings increasingly often utilise neural networks (NNs) for perception tasks such as localisation and object detection, \emph{one-sided neuro-symbolic POSGs (one-sided NS-POSGs)} were introduced \citep{RY-GS-GN-DP-MK:23-2}. In this model the agent with partial information observes the environment only through a trained NN classifier, and consequently the game is generalised to \emph{continuous environments}, to align with NN semantics, while  observations remain discrete. A point-based NS-HSVI method was developed to approximate values of one-sided NS-POSGs working with (polyhedral decompositions of) the continuous space.

Strategy synthesis for continuous games is more challenging than for the finite-state case~\citep{KH-BB-VK-CK:23}, since continuous-state spaces lead to an infinite number of strategies and discretisation suffers from  
the curse of dimensionality. 
Several \emph{offline} methods exist, based on counterfactual regret minimisation, heuristics or reinforcement learning (see Related Work, below).
In this paper, we consider \emph{online} methods, which can improve efficiency and adaptability.
The best performing online method~\citep{MM-MS-NB-VL-DM-et-al:17} continually resolves 
a local strategy 
that only keeps track of the agent's belief of its opponent state and a vector of 
opponent counterfactual values. 
Apart from \cite{KH-BB-VK-CK:23}, existing continual resolving approaches \citep{MM-MS-NB-VL-DM-et-al:17,MV-VK-VL:19,MS-MM-NB-RK-JD-others:23} are for extensive form games (EFGs), and cannot be directly applied to POSGs. This is because, although the two formalisms are connected \citep{VK-MS-NB-MB-VL:22},  transitioning between them is not straightforward. 

\startpara{Contributions}
We develop a \emph{continual resolving} approach for one-sided NS-POSGs,
addressing several challenges. Firstly, existing continual resolving approaches need to estimate the opponent's counterfactual values by solving a 
subgame at each stage \citep{MM-MS-NB-VL-DM-et-al:17}, which would be intractable for continuous games. 
Instead, for the agent with partial observation ($\agent_1$), we use the lower bound computed offline by NS-HSVI~\citep{RY-GS-GN-DP-MK:23-2}, giving a polyhedral bound
without solving a subgame. At each stage, we solve a linear program (LP), whose size is linear in the number of states in the current belief rather than the number of states reached, to compute the agent's action choice and update the lower bound. Thus, a \emph{stage strategy} is computed online, in the spirit of continual resolving,  for each  situation as it arises during execution, instead of storing a complete strategy. Although we require offline computation for NS-HSVI, existing continual resolving approaches need to train deep counterfactual value networks to solve the subgame. 
Importantly, our NS-HSVI continual resolving does not lose soundness, i.e., $\varepsilon$-exploitability \citep{NB-MJ-MB:14}.

We can use any synthesis method for fully-observable stochastic games 
to generate an $\varepsilon$-minimax strategy for the fully-informed agent ($\agent_2$)~\citep{arxiv}. However, 
solving the fully-observable case would generate a \emph{complete} strategy, which can be costly in terms of memory. 
We instead propose an online \emph{inferred-belief} strategy by observing that, by using the offline upper bound from NS-HSVI,
$\agent_2$ only needs to keep track of an inferred belief about $\agent_1\!$'s belief
and solve an LP,  linear in the number of states in the current belief, to synthesise an action and the next inferred-belief of $\agent_2$. 
Since $\agent_2$ is fully informed, it does not need to store its belief. 
This allows us to generate a simpler strategy than the complete strategy for $\agent_2$, which guarantees the value at the initial belief, known to both agents, but cannot optimally employ the suboptimal actions of $\agent_1$ during play.

Summarising the contribution, we present, for the first time, an efficient online method to compute an $\varepsilon$-minimax strategy profile for one-sided NS-POSGs, a variant of two-player continuous-state POSGs with neural perception mechanisms, 
by exploiting bounds pre-computed offline by a variant of HSVI.
We implement our approach, evaluate it on a pursuit-evasion model
inspired by mobile robotics and investigate the synthesised agent strategies. 

\startpara{Related Work} 
Existing offline
strategy synthesis methods for \emph{one-sided} POSGs include a space partition approach~\citep{WZ-TJ-HL:22},
a point-based approximate algorithm for continuous observations~\citep{WZ-TJ-HL:23},
projection to POMDPs based on factored representations~\citep{SC-NJ-SB-MS-UT:21} and HSVI algorithms for finite \citep{KH-BB-VK-CK:23} and continuous-state spaces \citep{RY-GS-GN-DP-MK:23-2}. Since POSGs and EFGs are connected through factored-observation stochastic games \citep{VK-MS-NB-MB-VL:22}, we next review relevant methods for two-agent zero-sum EFGs.

% CFR
Counterfactual Regret Minimization (CFR) \citep{MZ-MJ-MB-CP:07} 
exploits the fact that the time-averaged strategy profile of regret minimizing algorithms converges to an 
$\varepsilon$-minimax strategy profile, in two-agent zero-sum EFGs with imperfect information. Since its introduction, a variety of CFR variants have been proposed and successfully applied to games~\citep{ML-KW-MZ-MB:09,VL-ML-MB:15,NB-MJ-MB:14}.   
However, since the CFR-based approaches require iterative traversal of the game tree, they become intractable when the tree is large. Additionally, these approaches are \emph{offline} algorithms, returning a complete solution strategy that is difficult to represent and store. 

A number of algorithms based on game-theoretic learning models such as reinforcement learning and heuristic search have also been proposed, which are able to compute strategies for two-agent zero-sum games with imperfect information, 
including~\citep{BB-CK-VL-MP:14,JH-ML-DS:15,ML-VZ-AG-AL-KT-others:17,SM-JBL-KAW-PB-RF:21}  
and heuristic search value iteration (HSVI) \citep{RY-GS-GN-DP-MK:23-2,AD-OB-JSD-AS:23}, which we utilise in our work. However, these approaches are also offline algorithms and unable to refine strategies at test time.

The most related approach is
\emph{continual resolving} used in \cite{KH-BB-VK-CK:23}, which is based on DeepStack \citep{MM-MS-NB-VL-DM-et-al:17}, %for poker, 
although 
other variants have also been proposed, e.g.,~\citep{MV-VK-VL:19,MS-MM-NB-RK-JD-others:23,NB-AB-AL-QG:20}.
Both \cite{KH-BB-VK-CK:23} and our work
are variants of continual resolving for one-sided POSGs, except we consider continuous-state spaces. Under the belief update in \cite{KH-BB-VK-CK:23},
the current state could be missed and the belief might 
be empty. 
We assume a uniform stage strategy  
to fix this issue and ensure that the true state is always in 
the current belief. The LP size in \cite{KH-BB-VK-CK:23} is fixed, 
while the LP in our case varies at each stage because of
no prior enumeration of all reachable states. 
\cite{KH-BB-VK-CK:23} performs the belief update before $\agent_1$ taking action and observing,  
whereas
we update the belief using the action and the next observation,
which results in a more accurate belief.

\iffalse
Both \cite{KH-BB-VK-CK:23} and the work presented in this paper are variants of continual resolving for one-sided POSGs. The difference arises in the fact that \cite{KH-BB-VK-CK:23} considers finite-state spaces while in our case we consider continuous-state spaces, as well as the belief update. $1)$ Since \cite{KH-BB-VK-CK:23} uses an assumed minimax stage strategy for the fully informed agent based on the lower bound to update $\agent_1$'s belief, the true state could be missed and the belief thus might be stuck in being empty. We assume a uniform stage strategy for the fully informed agent to fix this issue and ensure that the true state is always in the support of the current belief. $2)$ The size of LP in \cite{KH-BB-VK-CK:23} is fixed as the number of states is fixed, while the LP in our case varies at each stage because there is no prior enumeration of all reachable states, which could be large. $3)$ The belief update in \cite{KH-BB-VK-CK:23} is performed before taking its action and observing the next agent state, wasting the useful information. However, we update the belief after taking its action and observing the next agent state, thus resulting in a more accurate belief. This action-observation information can also help to cut down states that could not be reached and thus reduce the number of points in the belief and further the size of LP.
\fi
\section{Background} 

\noindent
We briefly review the model of \cite{RY-GS-GN-DP-MK:23-2}, 
which generalises \emph{one-sided POSGs}~\citep{KH-BB-VK-CK:23} to continuous-state spaces and allows neural perception mechanisms.
Let $\mathbb{P}(X)$ and $\mathbb{F}(X)$ denote the spaces of probability measures and functions on a Borel space $X$, respectively.

\startpara{One-sided NS-POSGs} 
A \emph{one-sided neuro-symbolic POSG (NS-POSG)} $\csg$ 
is a two-player zero-sum infinite-horizon game with discrete actions and observations, where 
one player ($\agent_1$) is partially informed and the other ($\agent_2$) is fully informed. 
Unlike~\cite{KH-BB-VK-CK:23}, the game is played in 
a closed continuous environment $S_E$, which $\agent_1$ perceives only using 
perception function $\obs_1$ given as
a (trained) ReLU NN classifier that maps environment states to so called \emph{percepts}, ranging over a finite set $\Per_1$. The use of classifiers is aligned with, e.g., object detection or vision tasks in autonomous systems. 
We further assume that $\agent_1$ has a discrete local state space $\Loc_1$, which is observable to both agents, and that $\agent_2$ has full knowledge of the environment's state.

A game $\csg$ comprises
agents $\agent_1  {=} $ $(S_1,A_1,\obs_1,\delta_1)$, $\agent_2 {=} (A_2)$ and environment $E {=} (S_E,\delta_E)$, where $S_1 = \Loc_1 {\times} \Per_1$; 
$A = A_1 {\times} A_2$ are joint actions; $\obs_1 : (\Loc_1 {\times} S_E) \to \Per_1$ is $\agent_1\!$'s perception function (note that we allow NNs to additionally depend on local states); $\delta_1: (S_1 {\times} A) \to \mathbb{P}(\Loc_1)$ is $\agent_1\!$'s 
local transition function; and $\delta_E: (\Loc_1 {\times} S_E {\times} A) \to \mathbb{P}(S_E)$ is $E$'s finite-branching 
transition function.
We work in the \emph{belief space} $S_B \subseteq \mathbb{P}(S)$, where $S = S_1 {\times} S_E$, and assume an initial belief $b^{\mathit{init}}$
using the particle-based representation \citep{JMP-NV-MTS-PP:06,AD-NDF-NJG:01}. A belief of $\agent_1$ is given by $ b = (s_1, b_1)$, where $s_1 \in S_1$, $b_1 \in \mathbb{P}(S_E)$
and $b_1$ is represented by a weighted particle set $\{ (s_E^i, \kappa_i) \}_{i=1}^{N_b}$ where $\kappa_i \ge 0$ and $\mbox{$\sum\nolimits_{i=1}^{N_b}$} \kappa_i = 1$.  

The game starts in a state $s=(s_1,s_E)$, where $s_1=(\loc_1,  \per_1) \in S_1$, 
and $s$ is sampled from $b^{\mathit{init}}$. 
At each \emph{stage} of the game, 
both agents concurrently choose one of their actions.
If $a=(a_1,a_2)\in A$ is chosen,
the local state $\loc_1$ of $\agent_1$ 
is updated to $\loc_1'\in \Loc_1$ via $\delta_1(s_1,a)$, while the environment 
updates its state to $s_E'\in S_E$ via $\delta_E(\loc_1, s_E,a)$. Finally, $\agent_1$, based on $\loc_1'$,
generates its percept $\per_1' = \obs_1(\loc'_1,s'_E)$ at $s_E'$ and $\csg$ reaches the state $s'=((\loc_1', \per_1'), s_E')$. The probability of transitioning from $s$ to $s'$ under $a$ is $\delta(s, a)(s') = \delta_1(s_1,a)(\loc_1') \delta_E(\loc_1, s_E,a)(s_E')$.

\startpara{Strategies}
We distinguish between a \emph{history} $\pi$ (a sequence of states and joint actions,
where $\pi(k)$ is the $(k{+}1)$th state,
and $\pi[k]$ is the $(k{+}1)$th action) and 
a (local) \emph{action-observation history (AOH)} for $\agent_i$ (a sequence of its observations and actions).
For the fully-informed $\agent_2$, an AOH is a history. 
A \emph{strategy} of $\agent_i$ is a mapping 
$\sigma_i: \fpaths_{\csg,i} \to \mathbb{P}(A_i)$, where $\fpaths_{\csg,i}$ is the set of $\agent_i$'s finite AOHs. 
A \emph{(strategy) profile}  $\sigma=(\sigma_1, \sigma_2)$ is a pair of strategies 
and we denote by $\Sigma_i$ and  $\Sigma$ the sets of strategies of $\agent_i$ and profiles.
The choices for the players after a history $\pi$ are given by \emph{stage strategies}:
for $\agent_1$ this is a distribution $u_1 \in \mathbb{P}(A_1)$ and for $\agent_2$  a function $u_2: S \to \mathbb{P}(A_2)$, i.e., 
$u_2 \in \mathbb{P}(A_2 \mid S)$. 
Given a belief $(s_1, b_1)$, if $\agent_1$ chooses $a_1$, \emph{assumes} $\agent_2$ 
chooses $u_2$ and observes $s_1'$, then the updated belief of $\agent_1$ via Bayesian inference is denoted $(s_1', b_1^{s_1,a_1,u_2,s_1'})$.

\startpara{Objectives and values} 
We focus on infinite-horizon expected accumulated reward
$\mathbb{E}_b^{\sigma} [Y]$ when starting from $b$
under $\sigma$, where $Y(\pi) = \sum_{k=0}^{\infty} \beta^k r(\pi(k), \pi[k])$ for an infinite history $\pi$, 
 reward structure $r : (S {\times} A) \to \mathbb{R}$ and discount $\beta \in (0, 1)$, and $\agent_1$ and $\agent_2$ maximise and minimise the expected value.
Given $\varepsilon {\geq} 0$, a profile $\sigma^{\star} = (\sigma_1^{\star}, \sigma_2^{\star})$ is an \emph{$\varepsilon$-minimax strategy profile} if 
for any $b \in S_B$, $\smash{\mathbb{E}_b^{\sigma^{\star}}}[Y] \leq \mathbb{E}_b^{\sigma_1^{\star},\sigma_2}[Y] {+} \varepsilon$ for all $\sigma_2$ and $\mathbb{E}_b^{\sigma^{\star}}[Y] \ge \mathbb{E}_b^{\sigma_1,\sigma_2^{\star}}[Y] {-} \varepsilon$ for all $\sigma_1$. If $\varepsilon{=}0$, then  $\smash{\mathbb{E}_b^{\sigma^{\star}}}[Y]$ is the \emph{value} of $\csg$, denoted $V^\star$. 

\startpara{One-sided NS-HSVI} 
HSVI is an anytime algorithm that approximates the value $V^{\star}$ via \emph{lower} and \emph{upper bound} functions, updated through heuristically generated beliefs. 
One-sided NS-HSVI~\citep{RY-GS-GN-DP-MK:23-2} works with the continuous-state space of a one-sided NS-POSG using a generalisation of $\alpha$-functions, similar to~\cite{JMP-NV-MTS-PP:06}, except it uses \emph{polyhedral} representations induced from NNs instead of Gaussian mixtures.  
For $\varepsilon > 0$, one-sided NS-HSVI returns lower and upper bound functions $V_{\mathit{lb}}^{\Gamma} , V_{\mathit{ub}}^{\Upsilon} \in \mathbb{F}(S_B)$ to approximate $V^{\star}$ such that $V_{\mathit{lb}}^{\Gamma}(b) \leq V^{\star}(b) \leq V_{\mathit{ub}}^{\Upsilon} (b)$ for all $b \in S_B$ and $V_{\mathit{ub}}^{\Upsilon} (b^{\mathit{init}}) - V_{\mathit{lb}}^{\Gamma}(b^{\mathit{init}}) \leq \varepsilon$. Given $f: S \to \mathbb{R}$ and belief $(s_1, b_1)$, let $\langle f, (s_1, b_1) \rangle = \, \mbox{$ \int_{s_E \in S_E}$} f(s_1, s_E) b_1(s_E) \textup{d}s_E $. The lower bound $V_{\mathit{lb}}^{\Gamma}$ is represented 
via a finite set $\Gamma \subseteq \mathbb{F}(S)$ of \emph{piecewise-constant (PWC)} $\alpha$-functions such that $V_{\mathit{lb}}^{\Gamma}(s_1, b_1) = \max_{\alpha \in \Gamma} \langle \alpha, (s_1, b_1) \rangle$. 
The upper bound $V_{\mathit{ub}}^{\Upsilon}$ is represented by a finite set of belief-value pairs $\Upsilon \subseteq S_B {\times} \mathbb{R}$
and computed via an LP.
\section{NS-HVSI Continual resolving}\label{sec:strategy-synthesis}

\noindent
\emph{Continual resolving}, e.g.,~\citep{MM-MS-NB-VL-DM-et-al:17}, is an online method for computing an $\varepsilon$-minimax strategy in two-player, zero-sum imperfect information EFGs;
it keeps track of an agent's  belief of its opponent state 
and opponent counterfactual values to build and solve a subgame to synthesise choices, without building a complete strategy.
It is \emph{sound}, in computing an $\varepsilon$-minimax strategy, but can be expensive 
as it needs to estimate opponent counterfactual values by traversing the game tree.

We now present a novel variant of continual resolving, which utilises the \emph{lower bound} function $V_{\mathit{lb}}^{\Gamma}$ 
computed by one-sided NS-HSVI to synthesise an $\varepsilon$-minimax strategy for $\agent_1$ that achieves the desired $\varepsilon$ distance to the value function  
at the initial belief. 
The method is efficient as it only requires solving a single LP at each stage. We first introduce the following minimax operator.

\begin{defi}[Minimax]\label{defi:minimax-operator}
The minimax operator $T : \mathbb{F}(S_B) {\rightarrow} \mathbb{F}(S_B)$ is given by:
\begin{align}
  [TV](s_1, b_1) &  = \smash{\max\nolimits_{u_1\in \mathbb{P}(A_1)} \min\nolimits_{u_2\in \mathbb{P}(A_2 \mid S)}  \mathbb{E}_{(s_1,b_1),u_1,u_2} [r(s,a)] } \nonumber \\
  & + \smash{ \beta \mbox{$\sum\nolimits_{(a_1,s_1') \in A_1 \times S_1}$} P(a_1, s_1' \mid (s_1, b_1), u_1, u_2 ) V(s'_1, b_1^{s_1,a_1, u_2, s_1'})} \label{eq:minimax-operator}
\end{align}
for $V \in \mathbb{F}(S_B)$ and $(s_1, b_1) \in S_B$, where $\mathbb{E}_{(s_1,b_1),u_1,u_2} [r(s,a)]$ is the expected value of $r$. 
\end{defi}

\begin{algorithm}[tb]
\caption{NS-HSVI continual resolving for $\agent_1\!$'s strategy via the lower bound}
\label{alg:strategy-agent-1}
\textbf{Input}: $(s_{1}^{\mathit{init}}, b_{1}^{\mathit{init}})$, a finite set of PWC functions $\Gamma \subseteq \mathbb{F}(S)$ from one-sided NS-HSVI
\begin{algorithmic}[1] %[1] enables line numbers
\State $\mathit{Resolve}_1((s_{1}^{\mathit{init}}, b_{1}^{\mathit{init}}), \alpha^{\mathit{init}})$ where  $\alpha^{\mathit{init}} = \arg\max_{\alpha \in \Gamma} \langle \alpha, (s_{1}^{\mathit{init}}, b_{1}^{\mathit{init}}) \rangle$
\Function{$\mathit{Resolve}_1$}{$(s_1, b_1),  \alpha_1$}
\State $(\overline{v}^\star, \overline{\lambda}_1^\star, \overline{p}_1^\star) \leftarrow$ solve the LP \eqref{eq:lower-bound-LP-bounded-by-alpha} at $(s_1, b_1)$
\State $u_1^{\mathit{lb}}(a_1) \leftarrow p^{\star a_1}$ for all $a_1 \in A_1$
\State sample and play $a_1 \sim u_1^{\mathit{lb}}$
\State $s_1' \leftarrow$ observed $\agent_1$'s agent state
\State $\alpha^{\star a_1, s_1'} \leftarrow \mbox{$\sum_{\alpha \in \Gamma}$} (\lambda^{\star a_1, s_1'}_{\alpha} / p^{\star a_1}) \alpha$, \quad $u_2^{\mathit{lb}} \leftarrow$ an assumed stage strategy for $\agent_2$
\State $\mathit{Resolve}_1((s_1', b_1^{s_1, a_1, u_2^{\mathit{lb}}, s_1'}), \alpha^{\star a_1, s_1'})$
\EndFunction
\end{algorithmic}
\end{algorithm}
\startpara{NS-HSVI continual resolving} Motivated by \cite[Section 9.2]{KH-BB-VK-CK:23}, our online game-playing algorithm \emph{NS-HSVI continual resolving}, see \algoref{alg:strategy-agent-1}, generates a strategy for $\agent_1$, denoted $\sigma_1^{\mathit{lb}}$, by using the HSVI lower bound instead of opponent counterfactual values used in \cite{MM-MS-NB-VL-DM-et-al:17}. Since we have pre-computed $V_{\mathit{lb}}^{\Gamma}$ offline, our NS-HSVI continual resolving 
only keeps track of a belief $(s_1, b_1)$ 
and a PWC function $\alpha_1$ in the the convex hull, $\textup{Conv}(\Gamma)$, of $\Gamma$. 
Importantly, $\Gamma$ can be used to compute an action to play and update the tracking information at each stage (a belief and a  PWC function) by solving the LP presented below, thus avoiding  
the need to estimate the opponent  counterfactual values.

\begin{defi}[Stage strategy]\label{defi:stage-strategy-agent-1}
    For $((s_1, b_1),  \alpha_1) \in S_B \times \textup{Conv}(\Gamma)$ where $b_1$ is represented by $\{ (s_E^i, \kappa_i) \}_{i=1}^{N_b}$, a stage strategy $u_1^{\mathit{lb}}$ for $\agent_1$ is such that $u_1^{\mathit{lb}}(a_1) = p^{\star a_1}$ for $a_1 \in A_1$, where $(v_{s_E^i}^{\star })_{i=1}^{N_b}$, $\smash{(\lambda^{\star a_1, s_1'}_{\alpha} )_{(a_1,s_1') \in A_1 \times S_1,\alpha \in \Gamma}}$ and $\smash{(p^{\star a_1})_{a_1 \in A_1}}$ is a solution to the following LP:
   {\rm  \begin{align}
        \lefteqn{\mbox{\rm maximise} \; \; 
        \smash{\mbox{$\sum_{i=1}^{N_b}$} \kappa_i v_{s_E^i}}  \;\; \mbox{\rm subject to}}
        \nonumber  \\
         v_{s_E^i} & \leq \mbox{$\sum_{a_1 \in A_1}$} p^{a_1} r((s_1, s_E^i),(a_1,a_2))  + \beta \mbox{$\sum_{(a_1, s_1') \in A_1 \times S_1, s_E' \in S_E}$}  \nonumber \\
         & \;\;\; \cdot \delta((s_1, s_E^i), (a_1, a_2))(s_1', s_E') \mbox{$\sum_{\alpha \in \Gamma} \lambda_{\alpha}^{a_1, s_1'}$} \alpha (s_1', s_E') \quad \mbox{for } 1 \leq i \leq N_b \mbox{ and } a_2 \in A_2 \nonumber \\
         v_{s_E^i} & \ge \alpha_1(s_1,s_E^i) \quad \mbox{for } 1 \leq i \leq N_b \nonumber \\
         \lambda^{a_1, s_1'}_{\alpha} & \ge 0 \quad \mbox{for } a_1 \in A_1, \; s_1' \in S_1 \mbox{ and } \alpha \in \Gamma \nonumber \\
          p^{a_1} & = \mbox{$\sum_{\alpha \in \Gamma}$} \lambda_{\alpha}^{a_1, s_1'} \quad \mbox{for } a_1 \in A_1 \mbox{ and } s_1' \in S_1  \nonumber \\
          \mbox{$ \sum_{a_1 \in A_1} $} p^{a_1} & = 1. \label{eq:lower-bound-LP-bounded-by-alpha}
    \end{align} }
\end{defi}
Compared with the LP in \cite{RY-GS-GN-DP-MK:23-2} for solving $[TV_{\mathit{lb}}^{\Gamma}](s_1, b_1)$, the LP  \eqref{eq:lower-bound-LP-bounded-by-alpha} includes the additional constraints $\smash{v_{s_E^i} \ge \alpha_1(s_1,s_E^i)}$ for $1 \leq i \leq N_b$ \cite[Section 9.2]{KH-BB-VK-CK:23} to refine the minimax stage strategy in $[TV_{\mathit{lb}}^{\Gamma}](s_1, b_1)$, such that
the lower bound from 
$\alpha_1$ can be kept as the game evolves, since multiple minimax stage strategies for $\agent_1$ may exist and some of them may deviate from 
$\alpha_1$.  

% Author: Till Tantau
% Source: The PGF/TikZ manual
\def\xcolorversion{2.00}
\def\xkeyvalversion{1.8}
\def\coloralpha{45}

\begin{figure}[t]
    \centering
    \begin{tikzpicture}[node distance=1.3cm,>=stealth',bend angle=45,auto]

  \tikzstyle{state}=[circle,thick,draw=blue!\coloralpha,fill=blue!20,minimum size=5mm]
  \tikzstyle{transition}=[rectangle,thick,draw=black!75,
  			  fill=black!20]

  \tikzstyle{red place}=[place,draw=red!75,fill=red!20]
  % \tikzstyle{transition}=[rectangle,thick,draw=black!75,
  % 			  fill=black!20,minimum size=4mm]

  \tikzstyle{every label}=[red]

  \begin{scope}
    % First belief
    \node [state] (p1) {} 
    [children are tokens]
     child {node [token, fill=blue!\coloralpha,  minimum size=1.5mm] {}};
    \node [state] (p2) [below of=p1, xshift=-5mm, yshift=5mm, inner sep=0, draw=red!\coloralpha,fill=red!20] {}
    [children are tokens]
     child {node [token, fill=red!\coloralpha,  minimum size=1.8mm] {}};
    \node [state] (p3) [below of=p1, xshift=5mm, yshift=5mm, inner sep=0] {}
    [children are tokens]
     child {node [token, fill=blue!\coloralpha, minimum size=3mm] {}};
    \node [state] (p11) [above of=p2, xshift=-5mm,yshift=-5mm] {} 
    [children are tokens]
     child {node [token, fill=blue!\coloralpha,  minimum size=2.5mm] {}};
    \node [transition] (a1) [below of=p2, xshift=1mm, inner sep=0.8mm, fill=white, draw=white] {\scriptsize$(s_1, b_1), \alpha_1$};
    \node [state] (p4) [right of=p3, xshift=15mm, yshift=5mm, inner sep=0, draw=orange!\coloralpha,fill=orange!20] {};
    \node [state] (p5) [right of=p4, xshift=2mm, yshift=3mm, inner sep=0] {}
    [children are tokens]
     child {node [token, fill=blue!\coloralpha, minimum size=2mm] {}};
    \node [state] (p6) [right of=p5, xshift=-3mm, yshift=3mm, inner sep=0] {}
    [children are tokens]
     child {node [token, fill=blue!\coloralpha, minimum size=2.5mm] {}};
    \node [state] (p7) [below of=p5, yshift=6mm, xshift=1.5mm, inner sep=0] {}
    [children are tokens]
     child {node [token, fill=blue!\coloralpha, minimum size=1.2mm] {}};
    \node [state] (p8) [below of=p6, yshift=2mm, xshift=1mm, inner sep=0] {}
    [children are tokens]
     child {node [token, fill=blue!\coloralpha, minimum size=2mm] {}};
    \node [state] (p9) [below of=p5, xshift=7mm, yshift=10mm, inner sep=0, draw=orange!\coloralpha,fill=orange!20] {}
    [children are tokens]
     child {node [token, fill=orange!\coloralpha, minimum size=3mm] {}};
    \node [transition] (lp1) 
    [right of=a1, xshift=7mm, inner sep=0.8mm] {\scriptsize LP: HSVI $\Gamma$} edge [pre] (a1);
    \node [transition] (b1) 
    [right of=lp1, xshift=4mm, inner sep=0.8mm] {\scriptsize \shortstack{belief \\ update}} edge [pre] (lp1);
    \node [transition] (a2) [right of=b1, xshift=7mm, inner sep=0.8mm, fill=white, draw=white] {\scriptsize$(s_1', b_1'), \alpha_1'$};
    % \node [transition] (hsvi) [below of=b1, inner sep=0.8mm, fill=white] {\scriptsize HSVI lower bound function $V_{\mathit{lb}}^{\Gamma}$};
    % \node [transition] (lp2) 
    % [right of=a2, xshift=4mm, inner sep=0.8mm] {\scriptsize LP} edge [pre] (a2);
    
    \draw[->] (a1) ++ (0, 5mm) [>=stealth'] -- +(0, 5mm) node[left, yshift=-2.5mm] {\scriptsize$(s_1, b_1)$};
    \draw[<-] (p4)++ (-5mm, 0) [thick] -- +(-17mm, 0);
    \draw[->] (p4)++ (5mm, 0) [thick] -- +(6mm, 0);
    \node (lp11) [right of=lp1, xshift=-3.5mm] {};
    \node (lp12) [right of=lp1, xshift=-2.8mm, yshift=8mm] {};
    \node (p4-1) at (p4) [xshift=-2mm, yshift=-3mm] {};
    \draw [->] plot [smooth] coordinates {(lp11) (lp12) (p4-1)} node [left, xshift=-3.8mm, yshift=-6mm] {\scriptsize$a_1 {\sim} u_1^{\mathit{lb}}$};
    \draw[->] (p4)++ (0, -4mm) -- +(0, -10mm) node [left,yshift=6mm, xshift=1mm] {\scriptsize$s_1'$};
    \draw[->] (b1)  ++ (4.4mm, 0mm) [>=stealth'] -- +(6.7mm, 0) node [above, xshift=-5mm] {};
    \draw[->] (a2) ++ (0, 5mm) [>=stealth'] -- +(0, 5mm) node[left, yshift=-2.5mm] {\scriptsize$(s_1', b_1')$};
    \node (lp13) [above of=lp1, yshift=10mm] {};
    \node (lp14) [above of=lp1, yshift=10mm, xshift=8mm] {};
    \node (p4-2) at (p4) [xshift=-2mm, yshift=3mm] {};
    \draw [->] plot [smooth] coordinates {(lp13) (lp14) (p4-2)} node [above,xshift=-9mm, yshift=1.5mm] {\scriptsize$a_2{\sim}u_2$};
    \draw[->,>=stealth'] (lp1) -- + (0, -7mm) -- node [above, yshift=-0.7mm, xshift=4mm] {\scriptsize$\alpha_1'$} +(37.2mm, -7mm) -- + (37.2mm, -4mm);
    % \node [place] (s)  [below of=p2,label=above:$s\le 3$] {};
    % \node [place] (c2) [below of=s]                       {};
    % \node [place,tokens=1] (w2) [below of=c2]                      {};

    % \node [transition] (e1) [left of=p2] {}
    %   edge [pre,bend left]                  (p1)
    %   edge [post,bend right]                (s)
    %   edge [post]                           (p2);

    % \node [transition] (e2) [left of=c2] {}
    %   edge [pre,bend right]                 (w2)
    %   edge [post,bend left]                 (s)
    %   edge [post]                           (c2);

    % \node [transition] (l1) [right of=p2] {}
    %   edge [pre]                            (p2)
    %   edge [pre,bend left]                  (s)
    %   edge [post,bend right] node[swap] {2} (p1);

    % \node [transition] (l2) [right of=c2] {}
    %   edge [pre]                            (c2)
    %   edge [pre,bend right]                 (s)
    %   edge [post,bend left]  node {2}       (w2);
  \end{scope}

  \begin{scope}
    % Second belief
    \node [state] (e1) [right of=p9, xshift=7mm, yshift=2mm] {} 
    [children are tokens]
     child {node [token, fill=blue!\coloralpha,  minimum size=2.5mm] {}};
     \node [state] (e2) [below of=e1, yshift=6mm, draw=red!\coloralpha,fill=red!20] {} 
    [children are tokens]
     child {node [token, fill=red!\coloralpha,  minimum size=1.5mm] {}};
    \node [state] (e1-1) [right of=e1, xshift=-7mm, yshift=-3mm] {} 
    [children are tokens]
     child {node [token, fill=blue!\coloralpha,  minimum size=3mm] {}};
    \node [transition] (a3) [below of=e2, xshift=1mm, inner sep=0.8mm, fill=white, draw=white] {\scriptsize$(s_1, b_1)$};
    \node [state] (e3) [right of=e2, xshift=23mm, yshift=5mm, inner sep=0, draw=orange!\coloralpha,fill=orange!20] {};
    \node [state] (e4) [right of=e3, xshift=-1mm, yshift=1mm, inner sep=0] {}
    [children are tokens]
     child {node [token, fill=blue!\coloralpha, minimum size=1.5mm] {}};
    \node [state] (e5) [right of=e4, xshift=-4mm, yshift=3mm, inner sep=0] {}
    [children are tokens]
     child {node [token, fill=blue!\coloralpha, minimum size=2mm] {}};
    \node [state] (e6) [below of=e4, xshift=6mm, yshift=10mm, inner sep=0, draw=orange!\coloralpha,fill=orange!20] {}
    [children are tokens]
     child {node [token, fill=orange!\coloralpha, minimum size=3.5mm] {}};
    \node [transition] (lp2) 
    [right of=a3, xshift=5mm, inner sep=0.8mm] {\scriptsize LP: HSVI $\Upsilon$} edge [pre] (a3);
    \node [transition] (b2) 
    [right of=lp2, xshift=4mm, inner sep=0.8mm] {\scriptsize \shortstack{belief \\ update}} edge [pre] (lp2);
    \node [transition] (a4) [right of=b2, xshift=4mm, inner sep=0.8mm, fill=white, draw=white] {\scriptsize$(s_1', b_1')$};

    \draw[->] (a3) ++ (0, 5mm) [>=stealth'] -- +(0, 5mm) node[left, yshift=-2.5mm] {\scriptsize$(s_1, b_1)$};
    \draw[<-] (e3)++ (-5mm, 0) [thick] -- +(-17mm, 0);
    \draw[->] (e3)++ (4mm, 0) [thick] -- +(4mm, 0);
    \node (lp21) [right of=lp2, xshift=-3.5mm] {};
    \node (lp22) [right of=lp2, xshift=-2.8mm, yshift=8mm] {};
    \node (e3-1) at (e3) [xshift=-2mm, yshift=-3mm] {};
    \draw [->] plot [smooth] coordinates {(lp21) (lp22) (e3-1)} node [left, xshift=-3.8mm, yshift=-6mm] {\scriptsize$a_2 {\sim} u_2^{\mathit{ub}}$};
    \node (lp23) [above of=lp2, yshift=10mm] {};
    \node (lp24) [above of=lp2, yshift=10mm, xshift=8mm] {};
    \node (e3-2) at (e3) [xshift=-2mm, yshift=3mm] {};
    \draw [->] plot [smooth] coordinates {(lp23) (lp24) (e3-2)} node [above,xshift=-10mm, yshift=1.5mm] {\scriptsize$a_1{\sim}u_1$};
    \draw[->] (e3)++ (0, -4mm) -- +(0, -10mm) node [left,yshift=6mm, xshift=1mm] {\scriptsize$s_1'$} node [right,yshift=5.5mm, xshift=-0.5mm] {\scriptsize$a_1$};
    \draw[->] (a4) ++ (0, 5mm) [>=stealth'] -- +(0, 7mm) node[left, yshift=-3.5mm] {\scriptsize$(s_1', b_1')$};
    \draw[->] (b2)  ++ (4.4mm, 0mm) [>=stealth'] -- +(6mm, 0) node [above, xshift=-5mm] {};
  \end{scope}

  \begin{pgfonlayer}{background}
    \node at (a1) {\includegraphics[width=.08\paperwidth, height=0.045\paperwidth]{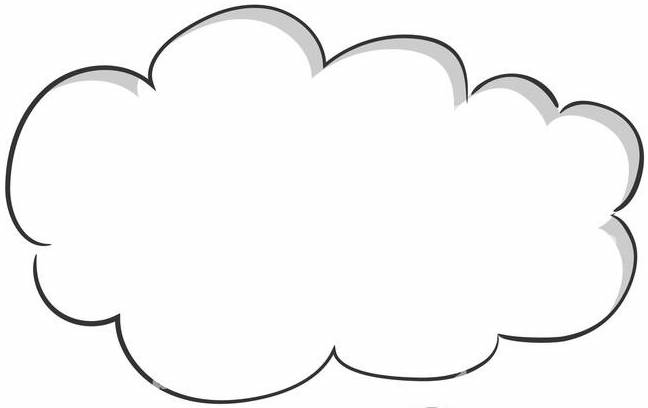}};
    \node at (a2) {\includegraphics[width=.08\paperwidth, height=0.045\paperwidth]{figures/strategy_figures/thought.jpg}};
    \node at (lp13) [xshift=-3mm] {\includegraphics[width=.022\paperwidth, height=0.02\paperwidth]{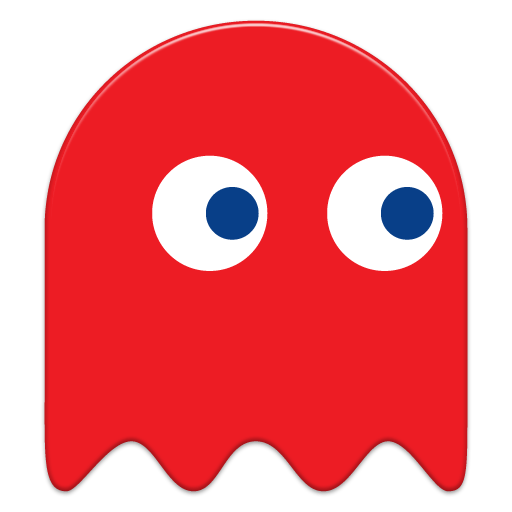}};
    \node at (a1) [yshift=-5mm] [xshift=-3mm] {\includegraphics[width=.02\paperwidth, height=0.02\paperwidth]{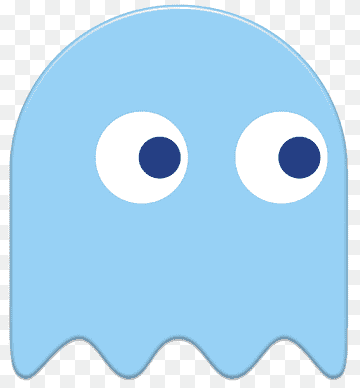}};
    \node at (a3) {\includegraphics[width=.06\paperwidth, height=0.045\paperwidth]{figures/strategy_figures/thought.jpg}};
    \node at (a4) {\includegraphics[width=.06\paperwidth, height=0.045\paperwidth]{figures/strategy_figures/thought.jpg}};
    \node at (a3) [yshift=-5mm] [xshift=-3mm] {\includegraphics[width=.022\paperwidth, height=0.02\paperwidth]{figures/strategy_figures/pacman-red.png}};
    \node at (lp23) [xshift=-3mm] {\includegraphics[width=.02\paperwidth, height=0.02\paperwidth]{figures/strategy_figures/pacman-blue-2.png}};
    \draw[dashed, join=round] (p11.north) ++ (-5.5mm, 4mm) -- ++ (78mm, 0mm) -- ++ (0mm, -36mm) -- ++ (-78mm, 0mm) -- ++ (0mm, 36mm)  {};
    \draw[dashed, join=round] (e1.north) ++ (-10.5mm, 5mm) -- ++ (71mm, 0mm) -- ++ (0mm, -36mm) -- ++ (-71mm, 0mm) -- ++ (0mm, 36mm)  {};
  \end{pgfonlayer}

\end{tikzpicture}
    \vspace*{-0.8cm}
    \caption{Left: NS-HSVI continual resolving for the partially-informed agent $\agent_1$ (blue). Right: inferred-belief strategy synthesis for the fully-informed agent $\agent_2$ (red).}\label{fig:strategy-diagram}
    \vspace*{-0.8cm}
\end{figure}

We illustrate
how the strategy $\sigma_1^{\mathit{lb}}$ is obtained in Fig.~\ref{fig:strategy-diagram} (left), where the red and orange circles indicate the current state, with the size of the interior solid circle representing the probability.  
When resolving the game locally $((s_1, b_1),  \alpha_1)$, $\agent_1$ (blue) plays an action $a_1$ sampled from a stage strategy $u_1^{\mathit{lb}}$ computed via \eqref{eq:lower-bound-LP-bounded-by-alpha}. Simultaneously, $\agent_2$ (red) plays an action $a_2$ sampled from a stage strategy $u_2$, which is unknown to $\agent_1$. The game moves to the next state and consequently $\agent_1$ observes a new agent state $s_1' \in S_1$. Based on $a_1$, $s_1'$ and an \emph{assumed} stage strategy $u_2^{\mathit{lb}}$ for $\agent_2$, see lines 7--8 of \algoref{alg:strategy-agent-1}, $\agent_1$ updates its belief to $(s_1', b_1^{s_1, a_1, u_2^{\mathit{lb}}, s_1'})$ via Bayesian inference, and generates $\alpha^{\star a_1, s_1'} \in \textup{Conv}(\Gamma)$ from a solution to \eqref{eq:lower-bound-LP-bounded-by-alpha}, which forms a new pair for the next resolving.
\begin{rek}\label{rek:two-requirements}
    There are two key properties of \algoref{alg:strategy-agent-1}. First, LP \eqref{eq:lower-bound-LP-bounded-by-alpha} admits at least one solution 
     as $V_{\mathit{lb}}^{\Gamma}$ is computed by the one-sided NS-HSVI. Second, the current state has to be in the support of the current belief, i.e., $\agent_1$ does not lose track of the current state. % and always puts positive probabilities over it. 
     Such a belief is called a proper belief and this is ensured by assuming the uniform stage strategy for $\agent_2$.
\end{rek}

\begin{lema}[Existence and proper belief]\label{lema:existence-proper-belief} For the NS-HSVI continual resolving at $((s_1, b_1),  \alpha_1)$, the LP \eqref{eq:lower-bound-LP-bounded-by-alpha} admits at least one solution, and if the current state is $(s_1, s_E)$, then $b_1(s_E) > 0$.
\end{lema}

\begin{proof}
     The optimal value $V^{\star}$ has lower and upper bounds $L = \min\nolimits_{s \in S, a \in A} r(s,a) /(1- \beta)$ and $
    U = \max\nolimits_{s \in  S, a \in A} r(s,a)/(1- \beta)$.
     Let $V_{\mathit{lb}}^{\Gamma'}$ be the lower bound of one-sided NS-HSVI~\citep{RY-GS-GN-DP-MK:23-2}. Existence follows if 
    \eqref{eq:lower-bound-LP-bounded-by-alpha} admits one solution after any point-based update. 
    The key is to check the feasibility of the first two constraints for $v_{s_E^i}$ of \eqref{eq:lower-bound-LP-bounded-by-alpha}. 
    Initially $\Gamma'= \{ \alpha_0 \}$ with $\alpha_0(s) = L$ for $s \in S$, and since $\alpha_1 = \alpha_0$ it is straightforward to verify that \eqref{eq:lower-bound-LP-bounded-by-alpha} admits at least one solution. 
    
    For the inductive step, we assume that \eqref{eq:lower-bound-LP-bounded-by-alpha} admits at least one solution for $\alpha \in \Gamma$ and $\alpha_1 \in \textup{Conv}(\Gamma)$. 
    The point-based update computes a new set $\Gamma' = \Gamma \cup \{ \alpha^{\star} \} $ of PWC functions, where \eqref{eq:lower-bound-LP-bounded-by-alpha} admits at least one solution for  $\alpha_1 = \alpha^{\star}$. Thus, \eqref{eq:lower-bound-LP-bounded-by-alpha} admits at least one solution for $\alpha \in \Gamma$ and $\alpha_1 \in \textup{Conv}(\Gamma) \cup \{ \alpha^{\star}\}$. 
    Following the proof of \cite[Lemma 9.6]{KH-BB-VK-CK:23}, we can show that \eqref{eq:lower-bound-LP-bounded-by-alpha} admits at least one solution for $\alpha \in \Gamma'$ and $\alpha_1 \in \textup{Conv}(\textup{Conv}(\Gamma) \cup \{ \alpha^{\star}\})$, i.e., $\alpha_1 \in \textup{Conv}(\Gamma')$.

    To show the belief is proper, since the initial state is sampled from $(s_{1}^{\mathit{init}}, b_{1}^{\mathit{init}})$, which is known to $\agent_1$, and the uniform stage strategy for $\agent_2$ is assumed, then the result follows.
\end{proof}
\noindent
We next show that our NS-HSVI continual resolving can inherit the soundness of continual resolving, i.e., it can compute an $\varepsilon$-minimax strategy for $\agent_1$.  

\begin{theorem}[$\varepsilon$-minimax strategy for $\agent_1$]\label{thom:strategy-agent-1}The strategy $\sigma_1^{\mathit{lb}}$ in \algoref{alg:strategy-agent-1} is an $\varepsilon$-minimax strategy for $\agent_1$  at $(s_{1}^{\mathit{init}}, b_{1}^{\mathit{init}})$, i.e., $\mathbb{E}_{(s_{1}^{\mathit{init}}, b_{1}^{\mathit{init}})}^{\sigma_1^{\mathit{lb}},\sigma_2} [Y] \ge V^{\star}(s_{1}^{\mathit{init}}, b_{1}^{\mathit{init}}) - \varepsilon$ for all $\sigma_2 \in \Sigma_2$.
\end{theorem}

\begin{proof}
    We adapt the proof presented for discrete one-sided POSGs in \cite[Proposition 9.7]{KH-BB-VK-CK:23}. 
    Let $\mathbb{E}_s^{\sigma} [Y]$ denote the expected value of $Y$
    when starting from  $s \in S$
    under $\sigma \in \Sigma$.
    Consider any $b = (s_1, b_1) \in S_B$ and $\alpha_1 \in \textup{Conv}(\Gamma)$. We assume that $\agent_1$ follows $\mathit{Resolve}_1$ in \algoref{alg:strategy-agent-1} at the first $t$ stages and then follows the uniform stage strategy.
    We denote such a strategy by $\sigma_1^{b, \alpha_1,t}$.
    \lemaref{lema:existence-proper-belief} guarantees that $\mathit{Resolve}_1$ can run for $t$ stages. We next prove that, for any $s_E$ with $b_1(s_E) > 0$, the expected value of $Y$ from $(s_1, s_E)$ under $\sigma_1^{b, \alpha_1,t}$ has the bound by $\alpha_1$:
    \begin{equation}\label{eq:expected-value-strategy-agent-1-lb}
        \smash{\mathbb{E}_{(s_1, s_E)}^{\sigma_1^{b, \alpha_1,t},\sigma_2}[Y] \ge \alpha_1(s_1, s_E) - \beta^{t} ( U - L) \quad \textup{ for all } \sigma_2 \in \Sigma_2.}
    \end{equation}
    We prove \eqref{eq:expected-value-strategy-agent-1-lb} by induction on $t \in \mathbb{N}$. For $t{=}0$, 
    \eqref{eq:expected-value-strategy-agent-1-lb} holds as $U$ and $L$ are trivial lower and upper bounds, 
    and $\alpha_1 \in \textup{Conv}(\Gamma)$ and $\alpha(s) \leq U$ for all $s \in S$ and $\alpha \in \Gamma$.
    We assume \eqref{eq:expected-value-strategy-agent-1-lb} holds for the first $t$ stages and any $b' = (s_1', b_1') \in S_B$. %, i.e., 
    The strategy $\smash{\sigma_1^{b, \alpha_1,t+1}}$ implies that, at $(s_1, b_1)$, $\agent_1$ takes 
    $u_1^{\mathit{lb}}$ (line 4) and then follows the strategy $\sigma_1^{b', \alpha',t}$ if $a_1$ is taken and $s_1'$ is observed, where 
    $b' = (s_1', b_1')$, $\smash{b_1' = b_1^{s_1, a_1, u_2^{\mathit{lb}}, s_1'}}$ and $\alpha' = \alpha^{\star a_1, s_1'}$. Letting $u_2 \in \mathbb{P}(A_2 \mid S)$ be the stage strategy of $\agent_2$ at $b$ given by $\sigma_2$, using $b_1(s_E) > 0$ by \lemaref{lema:existence-proper-belief}, the left side of \eqref{eq:expected-value-strategy-agent-1-lb} by replacing $t$ with $t+1$ equals: 
    \begin{align}
        % &   =    \smash{\mathbb{E}_{u_1^{\mathit{lb}},u_2} [r((s_1, s_E),a)]} + \beta \mbox{$\sum_{a_1,s_1', s_E'}$} P(a_1, s_1', s_E' \mid (s_1, s_E), u_1^{\mathit{lb}}, u_2 ) \smash{\mathbb{E}_{(s_1', s_E')}^{\sigma_1^{b', \alpha',t - 1},\sigma_2} [Y] } \nonumber \\
        \lefteqn{\mathbb{E}_{u_1^{\mathit{lb}},u_2} [r((s_1, s_E),a)] + \beta \mbox{$\sum_{(a_1,a_2) \in A \wedge (s_1', s_E') \in S}$} u_1^{\mathit{lb}}(a_1) u_2(a_2 \mid s_1, s_E)}  \nonumber \\
        & \qquad \cdot \delta((s_1, s_E), (a_1, a_2))(s_1', s_E') \mathbb{E}_{(s_1', s_E')}^{\sigma_1^{b', \alpha',t},\sigma_2} [Y] & \!\!\!\!\!\!\!\!\!\!\!\!\!\!\!\!\!\!\!\!\!\!\!\!\!\!\!\!\!\!\!\!\!\!\!\!\!\! \!\!\!\!\!\!\!\!\!\!\!\!\!\!\!\!\!\!\!\!\!\!\!\!\!\!\!\!\textup{by definition of $u_1^{\mathit{lb}}$, $u_2$ and $\delta$} \nonumber \\
        & \ge \min\nolimits_{a_2 \in A_2} \Big ( \mbox{$\sum_{a_1 \in A_1}$}  u_1^{\mathit{lb}}(a_1) r((s_1, s_E),(a_1, a_2)) + \beta \mbox{$\sum_{a_1 \in A_1 \wedge (s_1', s_E') \in S}$} u_1^{\mathit{lb}}(a_1)  & \nonumber \\
        & \qquad \cdot \delta((s_1, s_E), (a_1, a_2))(s_1', s_E') \mathbb{E}_{(s_1', s_E')}^{\sigma_1^{b', \alpha',t},\sigma_2} [Y] \Big) & \!\!\!\!\!\!\!\!\!\!\!\!\!\!\!\!\!\!\!\!\!\!\!\! \!\!\textup{by linearity in $u_2$} \nonumber \\
        & \ge \smash{\min\nolimits_{a_2 \in A_2} \big ( \mbox{$\sum_{a_1 \in A_1}$}  u_1^{\mathit{lb}}(a_1) r((s_1, s_E),(a_1, a_2)) + \beta \mbox{$\sum_{a_1 \in A_1 \wedge (s_1', s_E') \in S}$} u_1^{\mathit{lb}}(a_1)} & \nonumber \\
        & \qquad \cdot \delta((s_1, s_E), (a_1, a_2))(s_1', s_E') ( \alpha^{\star a_1, s_1'} (s_1', s_E') - \beta^{t}(U - L) ) \big) & \!\!\!\!\!\!\!\!\!\!\!\!\!\!\!\!\!\!\!\!\!\!\!\!\!\!\!\!\textup{by induction} \nonumber \\
        & = \smash{\min\nolimits_{a_2 \in A_2} \big ( \mbox{$\sum_{a_1 \in A_1}$}  u_1^{\mathit{lb}}(a_1) r((s_1, s_E),(a_1, a_2)) + \beta \mbox{$\sum_{a_1 \in A_1 \wedge (s_1', s_E') \in S}$} u_1^{\mathit{lb}}(a_1)} & \nonumber  \\
        & \qquad \cdot \delta((s_1, s_E), (a_1, a_2))(s_1', s_E') \alpha^{\star a_1, s_1'} (s_1', s_E') \big) - \beta^{t+1}(U - L) & \!\!\!\!\!\!\!\!\!\!\!\!\!\!\!\!\!\!\!\!\!\!\!\!\!\!\!\!\mbox{rearranging} \nonumber \\
        & \ge v^{\star}_{s_E} - \beta^{t+1}(U - L)  & \!\!\!\!\!\!\!\!\!\!\!\!\!\!\!\!\!\!\!\!\!\!\!\!\!\!\!\!\!\!\!\!\!\!\!\!\!\!\!\!\!\!\!\!\!\!\!\!\!\!\!\!\!\!\!\!\!\!\!\!\!\!\!\!\!\!\!\!\!\!\!\!\!\!\!\!\!\!\!\!\!\!\!\!\!\!\!\!\!\!\!\!\!\!\!\!\!\!\!\!\!\!\! \mbox{since $(\overline{v}^\star, \overline{\lambda}_1^\star, \overline{p}_1^\star)$ is a solution to \eqref{eq:lower-bound-LP-bounded-by-alpha} (first constraint)} \nonumber \\
        & \ge  \alpha_1(s_1, s_E) - \beta^{t+1}(U -L) & \!\!\!\!\!\!\!\!\!\!\!\!\!\!\!\!\!\!\!\!\!\!\!\!\!\!\!\!\!\!\!\!\!\!\!\!\!\!\!\!\!\!\!\!\!\!\!\!\!\!\!\!\!\!\!\!\!\!\!\!\!\!\!\!\!\!\!\!\!\!\!\!\!\!\!\!\!\!\!\!\!\!\!\!\!\!\!\!\!\!\!\!\!\!\!\!\!\!\!\!\! \textup{since $(\overline{v}^\star, \overline{\lambda}_1^\star, \overline{p}_1^\star)$ is a solution to \eqref{eq:lower-bound-LP-bounded-by-alpha} (second constraint)}\nonumber 
    \end{align}
    and hence \eqref{eq:expected-value-strategy-agent-1-lb} holds. Letting $\sigma_1^{\mathit{lb}} = \lim_{t \rightarrow \infty} \sigma_1^{b^{\mathit{init}}, \alpha^{\mathit{init}}, t}$
    by definition:
    \begin{align*}
        \mathbb{E}_{b^{\mathit{init}}}^{\sigma_1^{\mathit{lb}},\sigma_2} [Y] & = \mbox{$\int_{s_E \in S_E}$} b_1^{\mathit{init}} (s_E) \mathbb{E}_{(s_{1}^{\mathit{init}}, s_E)}^{\sigma_1^{\mathit{lb}},\sigma_2} [Y]  \textup{d} s_E \\
        & \ge \langle \alpha^{\mathit{init}}, (s_{1}^{\mathit{init}}, b_{1}^{\mathit{init}}) \rangle & \mbox{by \eqref{eq:expected-value-strategy-agent-1-lb} and definition of $\langle \cdot, \cdot \rangle$} \\
        & = V_{\mathit{lb}}^{\Gamma} (b^{\mathit{init}}) & \mbox{by line 1 of \algoref{alg:strategy-agent-1}} \\
        & \ge V^{\star}(b^{\mathit{init}}) - \varepsilon \nonumber & \!\!\!\!\!\!\!\!\!\!\mbox{since $V_{\mathit{lb}}^{\Gamma}$ is returned by one-sided NS-HSVI}
    \end{align*}
     which completes the proof.
\end{proof}

\vspace*{-0.4cm}
\section{Inferred-Belief Strategy Synthesis}\label{sec:strategy_agent_2}
\noindent
We complement our variant of continual resolving with strategy synthesis for $\agent_2$, which utilises the \emph{upper bound} function $V_{\mathit{ub}}^{\Upsilon}$ pre-computed offline and keeps track of an 
\emph{inferred} belief about what $\agent_1$ believes, which could differ from $\agent_1\!$'s true belief. 
Any offline method for fully-observable stochastic games could instead be used,  
with the associated high computational and storage cost of generating a complete strategy. Instead, we present an efficient \emph{online} 
algorithm, where only one LP is solved at each stage, to synthesise an $\varepsilon$-minimax strategy for $\agent_2$. Since we have pre-computed the offline upper bound, this strategy can guarantee the minimax value from the initial belief, which is known to both agents, but cannot optimally employ the suboptimal actions of $\agent_1$ during play. 

\begin{algorithm}[tb]
\caption{Inferred-belief strategy synthesis for $\agent_2$ via the upper bound}
\label{alg:strategy-agent-2}
\textbf{Input}: $(s_{1}^{\mathit{init}}, b_{1}^{\mathit{init}})$, a finite set of belief-value pairs $\Upsilon$ by one-sided NS-HSVI
\begin{algorithmic}[1] %[1] enables line numbers
\State $\mathit{Resolve}_2(s_{1}^{\mathit{init}}, b_{1}^{\mathit{init}})$
\Function{$\mathit{Resolve}_2$}{$s_1, b_1$}
\State $\smash{u_2^{\mathit{ub}} \leftarrow }$ $\agent_2\!$'s minimax strategy in $\smash{[TV_{\mathit{ub}}^{\Upsilon}](s_1, b_1)}$, \quad $(s_1, s_E) \leftarrow$ current observed state
\State sample and play $a_2 \sim u_2^{\mathit{ub}}(\cdot \mid s_1, s_E)$
\State $(a_1, s_1') \leftarrow $ $\agent_1\!$'s action and updated agent state
\State $\mathit{Resolve}_2(s_1', b_1^{s_1, a_1, u_2^{\mathit{ub}}, s_1'})$
\EndFunction
\end{algorithmic}
\end{algorithm}

\startpara{Inferred-belief strategy synthesis} Our inferred-belief strategy synthesis  for $\agent_2$ (\algoref{alg:strategy-agent-2}) returns a strategy  $\sigma_2^{\mathit{ub}}$ based on
$V_{\mathit{ub}}^{\Upsilon}$. The main idea of $\sigma_2^{\mathit{ub}}$ is that $\agent_2$ keeps a belief $(s_1, b_1)$, about $\agent_1\!$'s belief at the current stage, and then computes an action via $\mathit{Resolve}_2$
based on $(s_1, b_1)$ and $V_{\mathit{ub}}^{\Upsilon}$. This belief $(s_1, b_1)$ is \emph{inferred}, as $\agent_2$ has no access to what $\agent_1$ actually believes about the state, except the initial belief  which is common knowledge. However, since $\agent_2$ is fully-informed, it can simulate an inferred belief update of $\agent_1$, i.e., its belief about $\agent_1\!$'s next belief.

We illustrate the obtained strategy $\sigma_2^{\mathit{ub}}$ in Fig.~\ref{fig:strategy-diagram} (right). If $(s_1, b_1)$ is what $\agent_2$ believes about $\agent_1\!$'s belief 
and $(s_1, s_E)$ is the current state observed by $\agent_2$, then $\agent_2$ chooses $a_2$ sampled from the stage strategy $u_2^{\mathit{ub}}$ conditioned on $(s_1, s_E)$ in 
$[TV_{\mathit{ub}}^{\Upsilon}](s_1, b_1)$ computed via an LP \citep{RY-GS-GN-DP-MK:23-2}. At the same time, $\agent_1$ takes $a_1 \in A_1$ sampled from a stage strategy $u_1$, where $\agent_2$ does not know $u_1$. Then, the game moves to the next state and thus $\agent_1$ observes $s_1' \in S_1$. Based on $a_1$, $s_1'$ and  $u_2^{\mathit{ub}}$, $\agent_2$ updates its belief about $\agent_1\!$'s belief via Bayesian inference to $\smash{(s_1', b_1^{s_1, a_1, u_2^{\mathit{ub}}, s_1'})}$.

We next show that the inferred-belief strategy is sound, i.e., 
the inferred-belief carries enough information to generate  an $\varepsilon$-minimax strategy for $\agent_2$. 

\begin{lema}[Monotonicity]\label{lema:strategy-agent-2} If $V_{\mathit{ub}}^{\Upsilon}$ is an upper bound generated during the one-sided NS-HSVI, then $[TV_{\mathit{ub}}^{\Upsilon}](s_1, b_1) \leq  V_{\mathit{ub}}^{\Upsilon}(s_1, b_1)$ for all $(s_1, b_1) \in S_B$. 
\end{lema}
\begin{proof}
    For a given set $\Upsilon$ of belief-value pairs, we first show $[TV_{\mathit{ub}}^{\Upsilon}]$ is convex and continuous. From \cite[Theorem 6 and 7]{RY-GS-GN-DP-MK:23-2} 
    we have that
    $[T V_{\mathit{ub}}^{\Upsilon}](s_1, b_1)  = \mbox{$ \sup_{\alpha \in \Gamma^\Upsilon} $} \langle \alpha, (s_1,b _1) \rangle$,
    where 
    $\Gamma^\Upsilon \subseteq \mathbb{F}(S)$
    and $ L \leq \alpha(s) \leq U$ for all $s \in S$ and $\alpha \in \Gamma^\Upsilon$. By \cite[Proposition 4.9]{KH-BB-VK-CK:23}, it follows that $[TV_{\mathit{ub}}^{\Upsilon}](s_1, \cdot)$ is convex.
    For $b_1, b_1' \in \mathbb{P}(S_E)$ and $\alpha \in \Gamma^\Upsilon$, using \cite[Theorem 1]{RY-GS-GN-DP-MK:23-2}, we have $| \langle \alpha, (s_1, b_1) \rangle - \langle \alpha, (s_1, b_1') \rangle | \leq K_{\mathit{ub}} (b_1, b_1')$,
    where $K_{\mathit{ub}}$ measures belief difference, and hence
       $|  [TV_{\mathit{ub}}^{\Upsilon}](s_1, b_1) -  [TV_{\mathit{ub}}^{\Upsilon}](s_1, b_1') | \leq K_{\mathit{ub}}(b_1, b_1')$.
       
    Let $V_{\mathit{ub}}^{\Upsilon^t}$ and $I^t$ be the upper bound and the associated index set after the $t$-th point-based update,
    respectively. Similarly to \cite[Lemma 9.11]{KH-BB-VK-CK:23}, we can now prove the monotonicity of $[TV_{\mathit{ub}}^{\Upsilon^t}]$ by induction on $t \in \mathbb{N}$. 
    For $t{=}0$, since $\Upsilon^0 = \{ ((s_1^i, b_1^i), U) \in S_B {\times} \mathbb{R}  \mid i \in I^0 \}$ for some initial index set $I^0$. For any $(s_1, b_1) \in S_B$, using \eqref{eq:minimax-operator}, we have $[TV_{\mathit{ub}}^{\Upsilon^0}](s_1, b_1) \leq  \mbox{$\max_{s \in S \wedge a \in A}$} r(s, a) + \beta U  = (1 - \beta) U + \beta U = U = V_{\mathit{ub}}^{\Upsilon^0} (s_1, b_1)$.

    For the inductive step, we assume that $ [TV_{\mathit{ub}}^{\Upsilon^{t}}](s'_1, b'_1) \leq  V_{\mathit{ub}}^{\Upsilon^{t}}(s'_1, b'_1)$ for all $(s'_1, b'_1) \in S_B$. 
    Thus
    % we have 
    $y_i \ge V_{\mathit{ub}}^{\Upsilon^{t}}(s^i_1, b^i_1) \ge [TV_{\mathit{ub}}^{\Upsilon^{t}}](s^i_1, b^i_1)$ for $i \in I^t$.
    Let $(s_1, b_1) \in S_B$ be the belief for the $(t{+}1)$-th point-based update. By lines 8 and 9 of \cite[Algorithm 1]{RY-GS-GN-DP-MK:23-2}, we have $y^{\star} = [TV_{\mathit{ub}}^{\Upsilon^t}](s_1, b_1)$ and $\Upsilon^{t+1} = \Upsilon^t \cup \{ ((s_1, b_1), y^{\star}) \}$.
    Using \cite[Lemma 4]{RY-GS-GN-DP-MK:23-2}, we have $V_{\mathit{ub}}^{\Upsilon^t}(s'_1, b'_1) \ge V_{\mathit{ub}}^{\Upsilon^{t+1}} (s'_1, b'_1)$ for all $(s'_1, b'_1) \in S_B$, from which $[TV_{\mathit{ub}}^{\Upsilon^t}](s^i_1, b^i_1) \ge [TV_{\mathit{ub}}^{\Upsilon^{t+1}}](s^i_1, b^i_1)$ for any $i \in I^{t+1}$.
    Therefore we have that $y^i \ge  [TV_{\mathit{ub}}^{\Upsilon^{t}}](s^i_1, b^i_1) \ge [TV_{\mathit{ub}}^{\Upsilon^{t+1}}](s^i_1, b^i_1)$ for any $i \in I^{t+1}$.
    Now, for any $(s_1', b_1') \in S_B$, if $\smash{(\lambda_i^{\star} )_{i \in I^{t+1}_{s'_{\scale{.75}{1}}}}}$ is a solution for $V_{\mathit{ub}}^{\Upsilon^{t+1}}(s_1', b_1')$, then by construction:
    \begin{align}
       \lefteqn{\smash{V_{\mathit{ub}}^{\Upsilon^{t+1}}(s_1', b_1')}  = \smash{\mbox{$\sum\nolimits_{i \in I^{t+1}_{s'_{\scale{.75}{1}}}}$}\lambda_i^{\star} y_i + K_{\mathit{ub}}(b_1', \mbox{$\sum\nolimits_{i \in I^{t+1}_{s'_{\scale{.75}{1}}}}$} \lambda_i^{\star} b_1^i)}} & \nonumber \\
        & \ge \mbox{$\sum\nolimits_{i \in I^{t+1}_{s'_{\scale{.75}{1}}}}$}\lambda_i^{\star} [TV_{\mathit{ub}}^{\Upsilon^{t+1}}](s_1^i, b_1^i) + K_{\mathit{ub}}(b_1', \mbox{$\sum\nolimits_{i \in I^{t+1}_{s'_{\scale{.75}{1}}}}$} \lambda_i^{\star} b_1^i) & \mbox{by induction}  \nonumber \\
         & \ge \smash{[TV_{\mathit{ub}}^{\Upsilon^{t+1}}](s_1',  \mbox{$\sum\nolimits_{i \in I^{t+1}_{s'_{\scale{.75}{1}}}}$}\lambda_i^{\star}  b_1^i) + K_{\mathit{ub}}(b_1', \mbox{$\sum\nolimits_{i \in I^{t+1}_{s'_{\scale{.75}{1}}}}$} \lambda_i^{\star} b_1^i)} & \mbox{since $[TV_{\mathit{ub}}^{\Upsilon^{t+1}}]$ is convex} \nonumber \\
         & \ge [TV_{\mathit{ub}}^{\Upsilon^{t+1}}](s_1',  b_1') & \mbox{since $[TV_{\mathit{ub}}^{\Upsilon^{t+1}}]$ is $K_{\mathit{ub}}$-continuous} \nonumber 
    \end{align}
    and hence by induction $[TV_{\mathit{ub}}^{\Upsilon}]$ is monotone as required.
\end{proof}

\vspace{-5pt}

\begin{theorem}[$\varepsilon$-minimax strategy for $\agent_2$]\label{thom:strategy-agent-2}The strategy $\sigma_2^{\mathit{ub}}$ in \algoref{alg:strategy-agent-2} is an $\varepsilon$-minimax strategy  for $\agent_2$ at  $(s_{1}^{\mathit{init}}, b_{1}^{\mathit{init}})$, i.e., $\mathbb{E}_{(s_{1}^{\mathit{init}}, b_{1}^{\mathit{init}})}^{\sigma_1,\sigma_2^{\mathit{ub}}}[Y] \leq V^{\star}(s_{1}^{\mathit{init}}, b_{1}^{\mathit{init}}) + \varepsilon$ for all $\sigma_1 \in \Sigma_1$.
\end{theorem}
\begin{proof}
Consider $b {=} (s_1, b_1) \in S_B$. We assume that $\agent_2$ follows $\mathit{Resolve}_2$ in \algoref{alg:strategy-agent-2} for the first $t$ stages and then follows the uniform strategy,
and denote this strategy by $\smash{\sigma_2^{b,t}}$.
We prove by induction on $t \in \mathbb{N}$ that the expected value of $Y$ from $b$ under $\smash{\sigma_2^{b,t}}$ has the following upper bound:
 \begin{equation}\label{eq:expected-value-strategy-agent-2-ub}
        \smash{\mathbb{E}_{b}^{\sigma_1, \sigma_2^{b,t}}[Y]  \leq V_{\mathit{ub}}^{\Upsilon}(b) + \beta^{t} ( U - L)} \qquad \textup{ for all } \sigma_1 \in \Sigma_1.
    \end{equation}
   For $t=0$, the strategy $\smash{\sigma_2^{b,0}}$ implies that $\agent_2$ adopts the uniform strategy, %from the beginning, i.e., $(s_1,b_1)$, 
   and therefore \eqref{eq:expected-value-strategy-agent-2-ub} directly follows as $U$ and $L$ are  lower and upper bounds.
    
    For the inductive step, we assume \eqref{eq:expected-value-strategy-agent-2-ub} holds for the first $t$ stages. % and any $(\bar{s}_1, \bar{b}_1) \in S_B$. 
    The strategy $\sigma_2^{b,t+1}$ implies that at $b=(s_1, b_1)$, $\agent_2$ takes $u_2^{\mathit{ub}}$ (line 4) and then if $a_1$ is taken and $s_1'$ is observed, follows the strategy $\smash{\sigma_2^{b',t}}$, where $b' =(s_1',b_1')$ and $b_1' = b_1^{s_1, a_1, u_2^{\mathit{ub}}, s_1'}$. Letting $u_1 \in \mathbb{P}(A_1)$ be $\agent_1\!$'s stage strategy at $b$ given by any $\sigma_1$, the left-hand side of \eqref{eq:expected-value-strategy-agent-2-ub} by replacing $t$ with $t+1$ equals:
    \begin{align*}
           \lefteqn{\mathbb{E}_{b, u_1,u_2^{\mathit{ub}}} [r(s,a)] + \beta \mbox{$\sum_{(a_1,s_1') \in A_1 \times S_1}$} P(a_1, s_1' \mid b, u_1, u_2^{\mathit{ub}} ) \mathbb{E}_{b'}^{\sigma_1, \sigma_2^{b',t}} [Y]} \nonumber \\
        & \leq  \mathbb{E}_{b, u_1,u_2^{\mathit{ub}}} [r(s,a)] + \beta \mbox{$\sum_{(a_1,s_1')\in A_1 \times S_1}$} P(a_1, s_1' \mid b, u_1, u_2^{\mathit{ub}} ) (V_{\mathit{ub}}^{\Upsilon}(b') + \beta^{t} ( U - L)) &\!\!\!\!\!\!\!\!\!\!\!\!\!\!  \!\!\!\!\!\!\!\!\! \!\!\!\!\!\!\!\!\!\!\!\!\!\!\!\!\!\!\!\!\!\!\!\!\!\!  \mbox{by induction}  \nonumber \\
        &  \leq \mathbb{E}_{b, u_1^{\mathit{ub}},u_2^{\mathit{ub}}} [r(s,a)] + \beta \mbox{$\sum_{(a_1,s_1') \in A_1 \times S_1}$}  P(a_1, s_1' \mid b, u_1^{\mathit{ub}}, u_2^{\mathit{ub}} ) V_{\mathit{ub}}^{\Upsilon}(b') + \beta^{t+1} ( U - L) \nonumber \\
        &  &\!\!\!\!\!\!\!\!\!\!\!\!\!\!  \!\!\!\!\!\!\!\!\!\!\!\!\!\! \!\! \!\!\!\!\!\!\!\!\!\!\!\!\!\!\!\!\!\!\!\!\!\!\!\!\!\!\!\!\!\!\!\!\!\!\!\!\!\!\!\!\!\!\!\!\!\!\!\!\!\!\!\!\!\!\!\!\!\!\!\!\!\!\! \mbox{rearranging and since $u_1^{\mathit{ub}}$ is a minimax strategy}   \nonumber \\
        &  = [TV_{\mathit{ub}}^{\Upsilon}](b) + \beta^{t+1}(U - L) & \!\!\!\!\!\!\!\!\!\!\!\!\!\! \!\!\!\!\!\!\!\!\! \!\!\!\!\!\!\!\!\!\!\!\!\!\!\!\!\!\!\!\!\!\!\!\!\!\!\!\!\!\!\!\!\!\!  \mbox{by definition of $[TV_{\mathit{ub}}^{\Upsilon}]$} \nonumber \\
        & \leq V_{\mathit{ub}}^{\Upsilon} (b) + \beta^{t+1}(U - L)  & \!\!\!\!\!\!\!\!\!\!\!\!\!\!  \!\!\!\!\!\!\!\!\!\!\!\!\!\!\!\!\!\!\!\!\!\!\!\!\!\!\!\!\!\!\!\!\!\!\!\!\!\!\!\!\!\!\! \textup{by \lemaref{lema:strategy-agent-2}} 
    \end{align*}
    as required. Now, letting $\sigma_2^{\mathit{ub}} = \lim_{t \rightarrow \infty} \smash{\sigma_2^{b^{\mathit{init}}, t}}$, from \eqref{eq:expected-value-strategy-agent-2-ub} we have:
    \begin{align}
\smash{\mathbb{E}_{b^{\mathit{init}}}^{\sigma_1,\sigma_2^{\mathit{ub}}}[Y]  \leq V_{\mathit{ub}}^{\Upsilon} (b^{\mathit{init}}) \leq V^{\star}(b^{\mathit{init}}) + \varepsilon} \nonumber 
    \end{align}
    where the last inequality follows from the fact that $V_{\mathit{ub}}^{\Upsilon}$ is returned by one-sided NS-HSVI.
\end{proof}

\vspace{-12pt}
\begin{cor}[$\varepsilon$-minimax strategy profile]
    The profile $(\sigma_1^{\mathit{lb}}, \sigma_2^{\mathit{ub}})$ is an $\varepsilon$-minimax strategy profile.
\end{cor}
\section{Experiments}

\begin{figure}[t]
    \centering
    \includegraphics[width=0.2\textwidth]{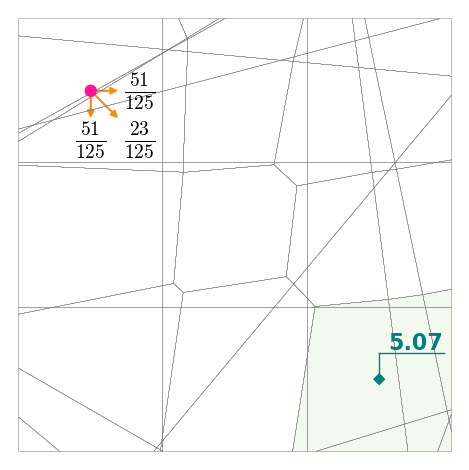}
    \hspace{-0.2cm}
    \includegraphics[width=0.2\textwidth]{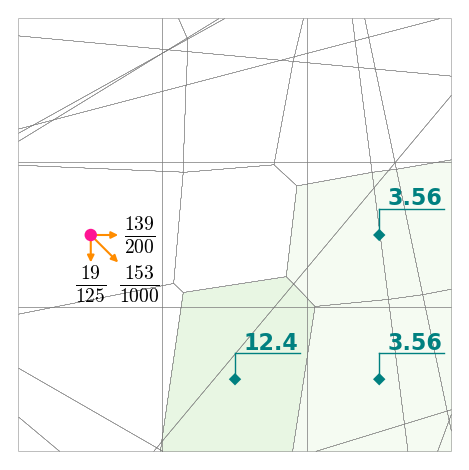}
    \hspace{-0.2cm}
    \includegraphics[width=0.2\textwidth]{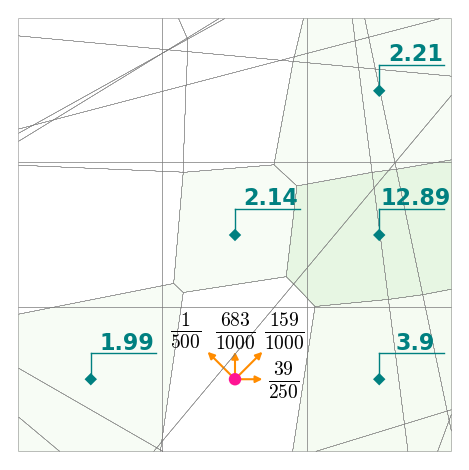}
    \hspace{-0.2cm}
    \includegraphics[width=0.2\textwidth]{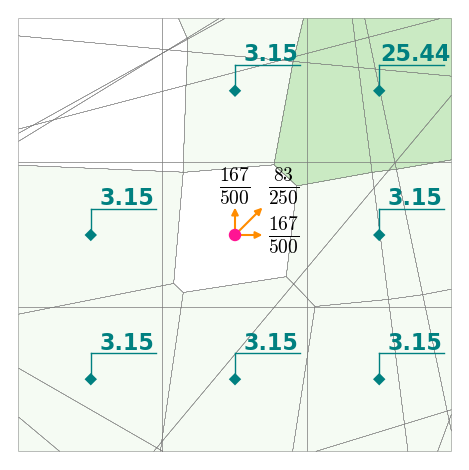}
    \hspace{-0.2cm}
    \includegraphics[width=0.2\textwidth]{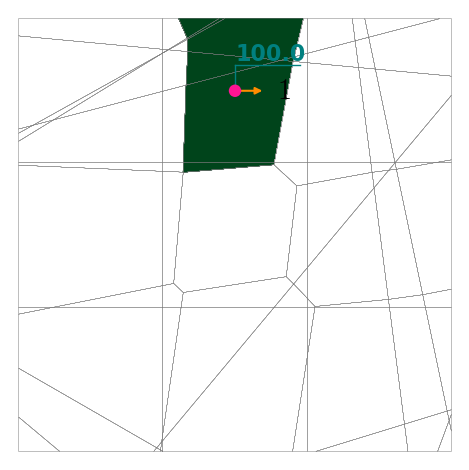}

    \includegraphics[width=0.2\textwidth]{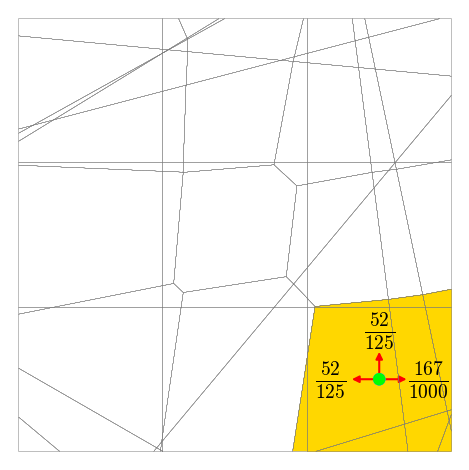}
    \hspace{-0.2cm}
    \includegraphics[width=0.2\textwidth]{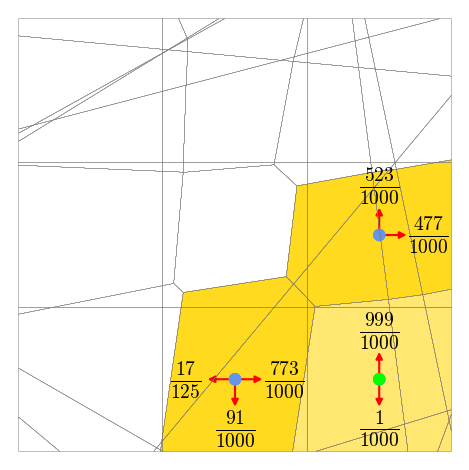}
    \hspace{-0.2cm}
    \includegraphics[width=0.2\textwidth]{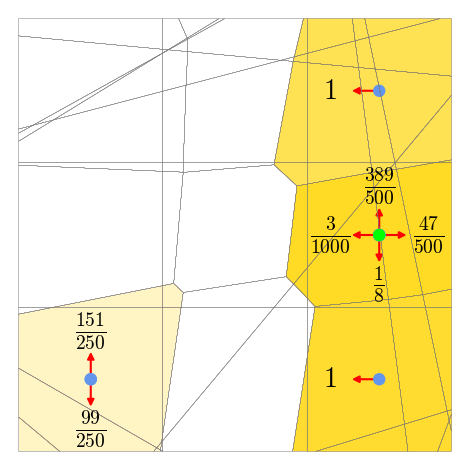}
    \hspace{-0.2cm}
    \includegraphics[width=0.2\textwidth]{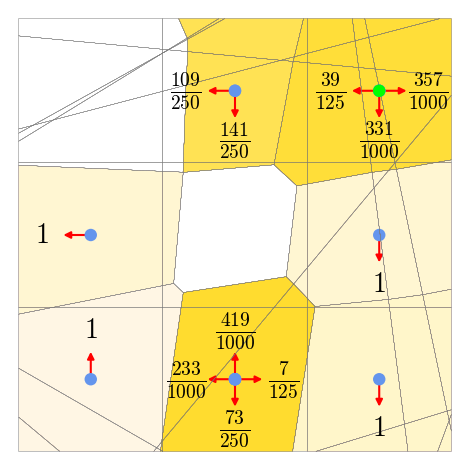}
    \hspace{-0.2cm}
    \includegraphics[width=0.2\textwidth]{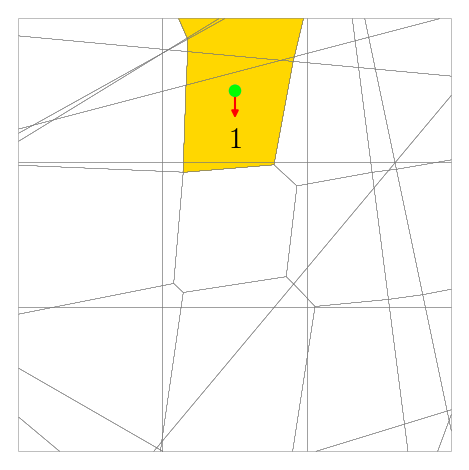}
    \vspace*{-0.2cm}
    \caption{Snippets of a synthesised strategy for the pursuer and evader
    % Snippets of a synthesised capturing strategy for the pursuer 
    (from left to right).}
    \label{fig:pursuit_evasion}
    \vspace*{-0.4cm}
\end{figure}

\noindent
We 
evaluated our method on a variant of a \emph{pursuit-evasion} game \citep{KH-BB-VK-CK:23}, inspired by mobile robotics \citep{THC-GAH-VI:11,VN08}; see the appendix of \cite{RY-GS-GN-DP-MK:23-2} for more detail. The game involves a \emph{pursuer}, whose aim is to capture an \emph{evader}.  
The pursuer is equipped with a ReLU NN classifier, which takes the location (coordinates) of the pursuer as
input and outputs one of the 9 abstract grid cells, each consisting of multiple polytopes, with the initial decomposition obtained by computing the preimage of the NN \citep{KM-FF:20}. The pursuer therefore observes
which cell it is in, but not its exact location, and knows neither the exact location nor cell of
the evader. 
The evader is fully informed
and knows the exact locations of both agents.
The evader is captured when both agents are in the same cell.
We set a discount of $0.7$, reward 100 for capture and timeout of 2$h$. The (offline) lower and upper bounds for the value of the initial belief are $5.0699$ and $6.0665$, respectively.
% [5.069905,6.066499]. 
We synthesise strategies, which demonstrate that
the pursuer can eventually capture the evader with positive probability.
Fig.~\ref{fig:pursuit_evasion} shows stage strategies and lower bounds of states in the belief of the pursuer (top), and evader's strategies and inferred beliefs (bottom), at different stages.
Lower bounds are coloured green and inferred beliefs yellow. The agent positions are highlighted (pink dot for pursuer and light green for evader). 
The belief of the pursuer and inferred-belief of the evader do not always coincide, 
e.g., in the third column, the state with bound $2.14$ is in the pursuer's belief, but not the inferred belief. 
 We observe that the pursuer's strategy selects the moves according to the magnitude of the bound, 
e.g., in the fourth column, the pursuer moves up or right, since the top right evader position has the highest %largest 
bound.

\vspace*{-0.4cm}
\section{Conclusions}

\noindent
We have developed an efficient online method to synthesise strategies for a variant of one-sided continuous-state POSGs with discrete observations and validated it on a pursuit-evasion game, in which the partially-informed agent uses a neural network for perception. We have shown that combining continual resolving, inferred beliefs and HSVI bounds computed offline can generate an $\varepsilon$-minimax strategy profile online.
% \marta{Some insight?} \rui{insights could be showing that combined with continual resolving and inferred beliefs, HSVI bounds computed offline are enough to generate an $\varepsilon$-minimax strategy profile online} 
% As future work, 
For future work, we will consider aggressive assumed stage strategies for the fully-informed agent,
% which offer worse lower bound guarantees but better efficiency for long-run games, 
since uniform strategies may lead to a large number of states in the belief and consequently large LPs to solve. %\marta{TBC} \rui{Future work could consider aggressive assumed stage strategies for the fully informed agent in the NS-HSVI continual resolving with worse lower bound guarantee but better efficiency for long-run games, instead of the uniform stage strategy which may lead to a large number of points in the belief after many stages and thus large LPs to solve at each stage afterwards. }

% \vfill
\startpara{Acknowledgements}
This project was funded by the ERC under the European
Union’s Horizon 2020 research and innovation programme
(FUN2MODEL, grant agreement No.834115).
\bibliography{references}

\end{document}